\documentclass[letterpaper,11pt]{article}

\usepackage{amsmath}
\usepackage{amssymb}

\usepackage{epsfig}
\usepackage{times}
\usepackage{mathptm}
\usepackage{psfrag}

\usepackage{algorithmic}
\usepackage[section]{algorithm}

\newcommand{\swap}{\mathsf{swap}}
\newcommand{\anc}{\mathsf{ANC}}
\newcommand{\des}{\mathsf{DES}}
\newcommand{\leftptr}{\mathsf{LeftPtr}}
\newcommand{\rightptr}{\mathsf{RightPtr}}
\newcommand{\ord}{\mathsf{ord}}

\newcommand{\donotshow}[1]{}

\newcommand{\ignore}[1]{}

\newtheorem{theorem}{Theorem}
\newtheorem{lemma}{Lemma}

\newtheorem{definition}{Definition}
\newcommand{\qed}{\rule[-0.2ex]{0.3em}{1.4ex}}

\newtheorem{claim}{Claim}
\newcommand{\mbegin}{\{\ \ }
\newcommand{\mend}{\}}

\newlength{\mleftindent}
\newlength{\mindent}
\newlength{\mboxwidth}
\newcommand{\mincrement}{\addtolength{\mboxwidth}{-\mindent}}
\newcommand{\mdecrement}{\addtolength{\mboxwidth}{\mindent}}

\newlength{\preprogramskip}
\newlength{\postprogramskip}
\newlength{\mexpwidth}
\newlength{\mexpindent}
\newcommand{\indentafterkeyword}{\hspace*{0.5em}}

\newcommand{\mslifelse}[3]  %if is short else is long
{\setlength{\mexpwidth}{\mboxwidth}%
\settowidth{\mexpindent}{{\bf if\indentafterkeyword}}%
\addtolength{\mexpwidth}{-\mexpindent}%
{\bf if\indentafterkeyword}\parbox[t]{\mexpwidth}{#1}\\
\mincrement \mbegin \parbox[t]{\mboxwidth}{#2 \mend} \mdecrement \\
{\bf else} \\
\mincrement \mbegin \parbox[t]{\mboxwidth}{#3}\\
\mend \mdecrement
}

{\vspace{-0.3em}\begin{itemize}
\setlength{\itemsep}{0.2\itemsep}%
\setlength{\parskip}{0.2\parskip}%
\setlength{\topsep}{0.0\topsep}%
}
{\end{itemize}}

{\begin{itemize}
\setlength{\itemsep}{0.2\itemsep}%
\setlength{\parskip}{0.2\parskip}%
\setlength{\topsep}{0.0\topsep}%
}
{\end{itemize}}

{\vspace{-0.3em}\begin{description}
\setlength{\itemsep}{0.2\itemsep}%
\setlength{\parskip}{0.2\parskip}%
\setlength{\topsep}{0.0\topsep}%
}
{\end{description}}

{\begin{description}
\setlength{\itemsep}{0.2\itemsep}%
\setlength{\parskip}{0.2\parskip}%
\setlength{\topsep}{0.0\topsep}%
}
{\end{description}}

{\vspace{-0.3em}\begin{enumerate}
\setlength{\itemsep}{0.2\itemsep}%
\setlength{\parskip}{0.2\parskip}%
\setlength{\topsep}{0.0\topsep}%
}
{\end{enumerate}}

{\begin{enumerate}
\setlength{\itemsep}{0.2\itemsep}%
\setlength{\parskip}{0.2\parskip}%
\setlength{\topsep}{0.0\topsep}%
}
{\end{enumerate}}

\newlength{\proofpostskipamount}
\newlength{\proofpreskipamount}
\newenvironment{proof}%
               {\par\vspace{\proofpreskipamount}\noindent{\bf Proof:}\hspace{0.5em}}% 0.5 before
               {\nopagebreak%
                \strut\nopagebreak%
                \hspace{\fill}\qed\par\vspace{\proofpostskipamount}\noindent}

               {\par\vspace{0.5ex}\noindent{\bf Proof #1:}\hspace{0.5em}}%
               {\nopagebreak%
                \strut\nopagebreak%
                \hspace{\fill}\qed\par\medskip\noindent}

\newlength{\mydefwidth}
\newlength{\mytextwidth}

\newcommand{\myurl}[1]{{\footnotesize \url{#1}}}

\title{\bf Faster Algorithms for Online Topological Ordering}
\author{Telikepalli Kavitha \hspace*{1.3cm} Rogers Mathew\\
\small{Indian Institute of Science}\\
\small{Bangalore, India}\\
\small{\sf \{kavitha,\ rogers\}@csa.iisc.ernet.in}}

\begin{document}
\date{}
\maketitle

\begin{abstract}
We present two algorithms for maintaining the topological order of
a directed acyclic graph with $n$ vertices, under an online edge insertion
sequence of $m$ edges. 
Efficient algorithms for online topological ordering have many applications,
including online cycle detection, which is to discover the first edge that introduces
a cycle under an arbitrary sequence of edge insertions in a directed graph.
The current fastest algorithms for the online topological ordering problem
run in time $O(\min(m^{3/2}\log n, m^{3/2}+n^2\log n))$ and $O(n^{2.75})$
(the latter algorithm is faster for dense graphs, i.e., when $m > n^{11/6}$).
In this paper we present faster algorithms for this problem.

We first present a simple algorithm with running time $O(n^{5/2})$ for 
the online topological ordering problem. This is the current fastest
algorithm for this problem on dense graphs, i.e., when $m > n^{5/3}$.
We then present an algorithm with running time
$O((m + n\log n)\sqrt{m})$, which is an improvement over the
$O(\min(m^{3/2}\log n, m^{3/2}+n^2\log n))$ algorithm - it is a {\em strict}
improvement when $m$ is 
sandwiched between $\omega(n)$ and $O(n^{4/3})$.
Our results yield an improved upper bound of $O(\min(n^{5/2}, (m + n\log n)\sqrt{m}))$
for the online topological ordering problem.
\end{abstract}

%\newpage

\section{Introduction}
\label{intro}
Let $G=(V,E)$ be a directed acyclic graph (DAG) with $|V|= n$ and $|E| = m$.
In a topological ordering, each vertex $v \in V$ is associated with a value
$\ord(v)$ such that for each directed edge $(u,v) \in E$ we have
$\ord(u) < \ord(v)$.
When the graph $G$ is known in advance (i.e., in an offline setting), there
exist well-known algorithms to compute a topological ordering of $G$
in $O(m+n)$ time \cite{CLR}.

In the {\em online} topological ordering problem, the edges of the 
graph $G$ are not
known in advance but are given one at a time. We are asked to maintain a
topological ordering of $G$ under these edge insertions. That is, each time
an edge is added to $G$,
we are required to update the function $\ord$ so that for all the edges
$(u,v)$ in $G$, it holds that $\ord(u) < \ord(v)$. 
The na\"ive way
of maintaining an online topological order, which is to compute the order each
time from scratch with the offline algorithm, takes $O(m^2 + mn)$ time.
However such an algorithm is too slow when the number of edges, $m$, is large.
Faster
algorithms are known for this problem (see Section~\ref{previous}). We show the
following results here\footnote{Here and in the rest of the paper, 
we make the usual assumption that $m$ is $\Omega(n)$ and there are no parallel
edges, so $m$ is $O(n^2)$.}.

\begin{theorem}
\label{thm1}
An online topological ordering of a directed acyclic graph $G$ on $n$
vertices, under a sequence of arbitrary edge insertions, can
be computed in time $O(n^{5/2})$, independent of the number of edges inserted.
\end{theorem}

\begin{theorem}
\label{thm2}
An online topological ordering of a directed acyclic graph $G$ on $n$
vertices, under any sequence of insertions of $m$ edges, can
be computed in time $O((m + n\log n)\sqrt{m})$.
\end{theorem}

The online topological ordering problem has several applications and
efficient algorithms for this problem are 
used in online cycle detection routine in pointer analysis \cite{PKH} and
in incremental evaluation of computational circuits \cite{AHRSZ90}. 
This problem
has also been studied
in the context of compilation \cite{MNR93,OLB92} where dependencies between
modules are maintained to reduce the amount of recompilation performed when
an update occurs.

\subsection{Previous Results}
\label{previous}
The online topological ordering problem
has been well-studied. Marchetti-Spaccamela et al. \cite{MNR} gave an algorithm
that can insert $m$ edges in $O(mn)$ time. Alpern et al. \cite{AHRSZ90}
proposed a 
different algorithm which runs in $O(\parallel\delta\parallel\log\parallel\delta\parallel)$ time per
edge insertion with $\parallel\delta\parallel$ measuring the number of edges of the
minimal vertex subgraph that needs to be updated. However, note that not all the
edges of this subgraph need to be updated and hence even 
$\parallel\delta\parallel$ time
per insertion is not optimal. Katriel and Bodlaender \cite{KB05} 
analyzed a variant of the algorithm in \cite{AHRSZ90} 
and obtained an upper bound of 
$O(\min\{m^{3/2}\log n, \ m^{3/2} + n^2\log n\})$ for a general DAG. In
addition, they show that their algorithm runs in time $O(mk\log^2n)$ for
a DAG for which the underlying undirected graph has a treewidth of $k$, and
they show an optimal running time of $O(n\log n)$ for trees.
Pearce and Kelly \cite{PK04} 
present an algorithm that empirically performs very well on
sparse random DAGs, although its worst case running time is inferior to 
\cite{KB05}.

Ajwani et al. \cite{AFM06} gave the first $o(n^3)$ algorithm for the
online topological ordering problem. They propose a simple algorithm that
works in time $O(n^{2.75}\sqrt{\log n})$ and $O(n^2)$ space, thereby
improving upon the algorithm \cite{KB05} for dense DAGs. With some simple 
modifications in their data structure, they get $O(n^{2.75})$ time
and $O(n^{2.25})$ space. They also demonstrate empirically that their algorithm
outperforms the algorithms in \cite{PK04,MNR,AHRSZ90} 
on a certain class of hard sequence of edge insertions.

The only non-trivial lower bound for online topological ordering is due to
Ramalingam and Reps \cite{RR94}, who showed that an adversary can force any 
algorithm to perform $\Omega(n\log n)$ vertex relabeling operations while 
inserting $n-1$ edges (creating a chain). There is a large gap between 
the lower bound of $\Omega(m + n\log n)$ and the upper bound of
$O(\min\{n^{2.75}, \ m^{3/2}\log n, \ m^{3/2} + n^2\log n\})$.
%\subsubsection{Our Results.}We present a simple $O(n^{5/2})$ algorithm
%for the online topological ordering problem on a DAG with $n$ vertices.
%Our $O(n^{5/2})$ algorithm is 
%always faster than the $O(n^{2.75})$ algorithm in \cite{} 
%and outperforms the algorithm in \cite{} whenever the 
%number of edges inserted is $\Omega(n^{5/3})$.

\paragraph{Our Results.} The contributions of our paper are as follows:
\begin{itemize}
\item Theorem~\ref{thm1} shows an upper bound of $O(n^{5/2})$ for
the online topological ordering problem. This is
always better than the previous best upper bound of $O(n^{2.75})$ in \cite{AFM06}
for dense graphs. Our $O(n^{5/2})$ algorithm is the current fastest algorithm for
online topological ordering when
$m > n^{5/3}$.

\item Theorem~\ref{thm2} shows 
another improved upper bound of $O((m+n\log n)\sqrt{m})$.
This improves upon the bounds of 
$O(m^{3/2}\log n)$ and $O(m^{3/2}+n^2\log n)$ given in 
\cite{KB05}. Note that this is a strict improvement over 
$\Theta(\min(m^{3/2}\log n $, $(m^{3/2}+n^2\log n)))$ when $m$ is
sandwiched between $\omega(n)$ and $O(n^{4/3})$.
\end{itemize}
Combining our two algorithms, we have an improved
upper bound of $O(\min(n^{5/2}, (m+n\log n)\sqrt{m}))$ for the
online topological ordering problem.

Our $O(n^{5/2})$ algorithm is very simple and basically involves traversing 
successive locations of an array and checking the entries of the
adjacency matrix; it
uses no special data structures and is easy to implement, 
%Its description and
%analysis take into account cycle detection also, 
so it would be an efficient online cycle detection subroutine
in practice also. The tricky part here is showing the bound on
its running time
and we use a result from \cite{AFM06} in its analysis.
Our $O(m+n\log n)\sqrt{m})$ algorithm is an adaptation of the
Katriel-Bodlaender algorithm in \cite{KB05} (in turn based on the
algorithm in \cite{AHRSZ90}) and uses the Dietz-Sleator ordered list
data structure
and Fibonacci heaps.

\paragraph{Organization of the paper.} In Section~\ref{algo} we
present our $O(n^{5/2})$ algorithm and show its correctness.
We analyze its running time in Section~\ref{runtime}. 
We present our $O((m+n\log n)\sqrt{m})$ algorithm and its analysis in 
Section~\ref{new-algo}. The missing details are given in the Appendix.

\section{The $O(n^{5/2})$ Algorithm}
\label{algo}
We have a directed acyclic graph $G$ on vertex set $V$. 
In this section we present an algorithm that maintains a bijection
$\ord: V \rightarrow \{1,2,\ldots,n\}$ which is our topological ordering.
Let us assume
that the graph is initially the empty graph and so any bijection from
$V$ to $\{1,2,\ldots,n\}$ is a valid topological ordering of $V$ at the onset. 
New edges get inserted to this graph and after each edge insertion, we
want to update the current bijection from $V$ to $\{1,2,\ldots,n\}$ 
to a valid topological
ordering.
 
Let the function $\ord$  from $V$ onto $\{1,2,\ldots,n\}$ denote our
topological ordering.
We also have the inverse function of $\ord$ stored as 
an array $A[1..n]$, where $A[i]$ is the vertex $x$ such that $\ord(x) = i$.
We have the 
$n \times n$ $0$-$1$ 
adjacency matrix $M$ of the directed acyclic graph, corresponding
to the edges inserted so far; $M[x,y] = 1$ if and only if there is an
edge directed from vertex $x$ to vertex $y$.

%When an existing edge $(u,v)$ is deleted from $G$, all we need to do is 
%to change $M[u,v]$ to $0$.
When a new edge $(u,v)$ is added to $G$, there are two
cases: (i) either $\ord(u) < \ord(v)$ in which case the current ordering 
$\ord$ is
still a valid ordering, so we need to do nothing, except set the entry
$M[u,v]$ to $1$, or (ii)~$\ord(u) > \ord(v)$ in which case we need to
update $\ord$. We now present 
our algorithm to update $\ord$, when an edge $(u,v)$ such that 
$\ord(u) > \ord(v)$, is 
added to $G$. If $(u,v)$ creates a cycle, the algorithm quits; else it updates
$\ord$ to a valid topological ordering. 

[Let $s \leadsto t$ indicate that there is a directed path (perhaps, of length
0) from vertex $s$ to vertex $t$ in $G$. If $s\leadsto t$, we say $s$ is an {\em ancestor}
of $t$ and $t$ is a {\em descendant} of $s$.
We use $s \rightarrow t$ to indicate that
$(s,t) \in E$.]

\subsection{Our algorithm for inserting a new edge $(u,v)$ where $\ord(u) > \ord(v)$}
\label{summary-algo}

Let $\ord(v) = i$ and $\ord(u) = j$.
%We can assume that $j > i+1$
%(otherwise $j = i+1$, so $u$ and $v$ occupy adjacent positions in $A$, and our algorithm
%just needs to swap $u$ and $v$ in $A$, i.e., make 
%$\ord(u) = i$ and $\ord(v) = j$ to
%get a valid topological ordering).
Our algorithm works only on the subarray $A[i..j]$ and 
computes a subset (call it $\anc$) of ancestors of $u$
in $A[i..j]$ and a subset (call it $\des$) of
descendants of $v$ in $A[i..j]$; it assigns new positions in $A$
to the vertices in $\anc \cup \des$.
This yields the new topological ordering $\ord$.
%Note that $\ord(x)$ for an element $x$
%gets changed only if $x$ occurs in the set $\anc \cup \des$ during
%the course of the algorithm. Otherwise, $\ord(x)$ is unchanged.

We describe our  algorithm in two phases: 
$\mathsf{Phase}~1$ and $\mathsf{Phase}~2$. 
In $\mathsf{Phase}~1$ we construct the sets:
\begin{eqnarray*}
   \des & = & \{y: i \le \ord(y) \le t \ \mathsf{and} \ v \leadsto y\} \ \text{(the set of descendants of $v$ in the subarray $A[i..t]$)} \\
   \anc & = & \{w: t \le \ord(w) \le j \ \mathsf{and} \ w \leadsto u\} \ \text{(the set of ancestors of $u$ in the subarray $A[t..j]$)}
\end{eqnarray*}
where $t$ is a number such that
$i \le t \le j$ and  $t$ has the following property:
{\em if $G$ is a DAG, then the number of descendants of $v$ in 
$A[i..t]$ is exactly equal to the number of ancestors of $u$ in $A[(t+1)..j]$.}

We then check if $(x,y) \in E$ for any $x \in \des$ and $y \in \anc$, or if
$A[t] \in \des\cap\anc$.
If either of these is true, then  
the edge $(u,v)$ creates a cycle and $G$ is no longer a DAG,
so our algorithm quits.
Else, we delete the elements of $\des\cup\anc$ from their locations in
$A$, thus creating empty locations in $A$.

$\mathsf{Phase}~2$ deals with inserting elements of $\anc$ in $A[i..t]$ and
the elements of $\des$ in $A[(t+1)..j]$.
Note that we cannot place the vertices in $\anc$ straightaway in the empty
locations previously occupied by $\des$ in $A[i..t]$ since there might be 
further ancestors of elements of $\anc$ in $A[i..t]$. Similarly, there might be
descendants of elements of
$\des$ in $A[(t+1)..j]$. Hence we need $\mathsf{Phase}~2$ to add more elements
to $\anc$ and to $\des$, and to
insert elements of $\anc$ in $A[i..t]$ and those of $\des$ in
$A[(t+1)..j]$ correctly.

\subsubsection{\textsc{Phase}~1.}
\label{label1}
We now describe $\mathsf{Phase}~1$ of 
our algorithm in detail. 
Initially the set $\anc = \{u\}$ and the set $\des = \{v\}$. We maintain 
$\anc$ and $\des$ as queues.
%When an edge $(u,v)$ where $\ord(u)> \ord(v)$ is inserted to the graph, 
%our goal is to find the ancestors
%of $u$ in $A[i..j]$ and the descendants of $v$ in $A[i..j]$. A forward
%depth first search from $v$ and a backward depth first search from $u$  would achieve this (this
%is the method used in \cite{PK04}); however such a step would be too 
%expensive. Also, it is
%Our goal is to find  ancestors
%of $u$ in $A[i..j]$ and descendants of $v$ in $A[i..j]$ and permute their
%locations.
%However we do not need to compute the entire set, $S_{anc}$,
%of ancestors of $v$ in $A[i..j]$ and
%the entire set, $S_{des}$, of descendants of $u$ in $A[i..j]$. We compute a 
%In $\mathsf{Phase}~1$ we compute a 
%subset $\des \subseteq S_{des}$ and a subset $\anc \subseteq S_{anc}$
%as follows.
\begin{figure}[h]
\psfrag{A:}{$A:$}
\psfrag{u}{$u$}
\psfrag{v}{$v$}
\psfrag{i}{$i$}
\psfrag{j}{$j$}
\psfrag{t}{$t$}
\psfrag{l}{$\leftptr$}
\psfrag{r}{$\rightptr$}
\centerline{\includegraphics[width=3.6in]{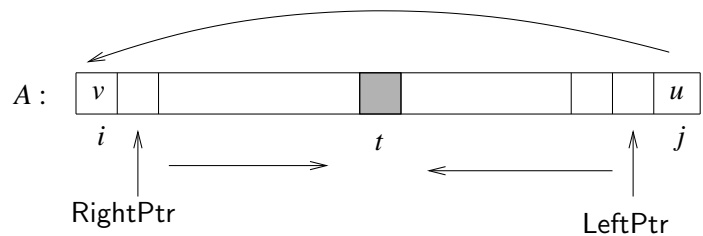}}
\caption{The pointer $\leftptr$ and $\rightptr$ meet at the location $t$.}
\label{figure2}
\end{figure}
We  move a pointer $\leftptr$ from location $j$ leftwards 
(towards $i$ as shown in Figure~\ref{figure2})  
in order to find an ancestor $w$ of $u$, which gets added to the end
of the queue $\anc$.
Then we move a pointer $\rightptr$ from location $i$ rightwards 
to find a descendant $y$ of $v$, which gets added to the end of $\des$.
Then we go back to $\leftptr$, and 
thus interleave adding a vertex to $\anc$ with adding
a vertex to $\des$ so that we balance the size of $\anc$ constructed so far
with the size of $\des$. When $\leftptr$ and $\rightptr$ meet, that
defines our desired location $t$. 
We present the detailed algorithm for $\mathsf{Phase}~1$ as 
Algorithm~\ref{fig-algo}. 
[If $(x,y)\in E$, we say $x$ is a {\em predecessor} of $y$
and $y$ is a {\em successor} of $x$.]

%\subsection{A detailed description for $\mathsf{Phase}~1$ of our $O(n^{5/2})$ algorithm (from Section~\ref{label1})}
\begin{algorithm}[h]
\begin{algorithmic}
\STATE{-- Initialize $\anc = \{u\}$ and $\des = \{v\}$.}
\STATE{-- set $\rightptr = i$ and $\leftptr = j$.}
\COMMENT{So $\leftptr > \rightptr$.}
\WHILE{$\mathsf{TRUE}$}
\STATE{$\leftptr = \leftptr - 1$;}
\WHILE{$\leftptr > \rightptr$ {\bf and} $A[\leftptr]$ is not a predecessor of any vertex in $\anc$}
\STATE{$\leftptr = \leftptr - 1$;}
\ENDWHILE
%\IF{$\leftptr = \rightptr$} 
\IF{$A[\leftptr]$ is a predecessor of some vertex in $\anc$}
\STATE{-- Insert the vertex $A[\leftptr]$ to the queue $\anc$.}
\ENDIF
\IF{$\leftptr = \rightptr$} 
\STATE{\bf break}
\COMMENT{this makes the algorithm break the while $\mathsf{TRUE}$ loop}
\ENDIF

%\COMMENT{Else $\leftptr > \rightptr$, so $A[\leftptr]$ has to be a predecessor of some vertex in $\anc$}
%\STATE{-- Insert the vertex $A[\leftptr]$ to the queue $\anc$.}

\vspace*{0.2cm}

\COMMENT{Now the symmetric process from the side of $v$.}

\vspace*{0.2cm}
\STATE{$\rightptr = \rightptr + 1$;}
\WHILE{$\rightptr < \leftptr$ {\bf and} $A[\rightptr]$ is not a successor of any vertex in $\des$}
\STATE{$\rightptr = \rightptr + 1$;}
\ENDWHILE
%\IF{$\rightptr = \leftptr$}
\IF{$A[\rightptr]$ is a successor of some vertex in $\des$}
\STATE{-- Insert the vertex $A[\rightptr]$ to the queue $\des$.}
\ENDIF
\IF{$\rightptr = \leftptr$}
\STATE{\bf break}
\ENDIF
\COMMENT{If the algorithm does not break the while $\mathsf{TRUE}$ loop, then $\leftptr > \rightptr$.}
%\COMMENT{Else $\rightptr < \leftptr$, so $A[\rightptr]$ has to be a successor of any vertex in $\des$}
%\STATE{-- Insert the vertex $A[\rightptr]$ to the queue $\des$.}
\ENDWHILE
\end{algorithmic}
\caption{Our algorithm to construct the sets $\anc$ and $\des$ in $\mathsf{Phase}~1$.}
\label{fig-algo}
\end{algorithm}

The above algorithm
terminates when  $\leftptr = \rightptr$ is satisfied.
Set $t$ to be this
location: that is, $t = \rightptr = \leftptr$.
It is easy to check that Algorithm~\ref{fig-algo} constructs the sets:
\[ \des = \{y: i \le \ord(y) \le t \ \mathsf{and} \ v \leadsto y\}; \ \ \ 
   \anc = \{w: t \le \ord(w) \le j \ \mathsf{and} \ w \leadsto u\}.
\]
That is, $\des$ is the set of descendants
of $v$ in $A[i..t]$ and $\anc$ is the set of ancestors of $u$
in $A[t..j]$. The following lemma is straightforward.

\begin{lemma}
\label{cycle-check}
If the new edge $(u,v)$ creates a cycle, then 
(i) either  $A[t] \in \des \cap \anc$ or
(ii) there is some $x \in \des$ and $y \in \anc$ such that
 there is an edge from $x$ to $y$.
\end{lemma}
\begin{proof}
If a cycle is created by the insertion of the edge $(u,v)$, then it
implies that $v \leadsto u$ in the current graph. That is,
there is a directed path $\rho$ in the 
graph from $v$ to $u$, before $(u,v)$ was inserted. 
This implies that there is either
an element in $\des\cap\anc$ or there is 
an edge in $\rho$ that connects a descendant of $v$ in 
$A[i..t]$ to an ancestor of $u$ in $A[t..j]$ (see Figure~\ref{figure3}).
\begin{figure}[h]
\psfrag{A:}{$A:$}
\psfrag{u}{$u$}
\psfrag{v}{$v$}
\psfrag{i}{$i$}
\psfrag{j}{$j$}
\psfrag{t}{$t$}
\psfrag{path}{The path $\rho$ from $v$ to $u$}
\centerline{\includegraphics[width=3.4in]{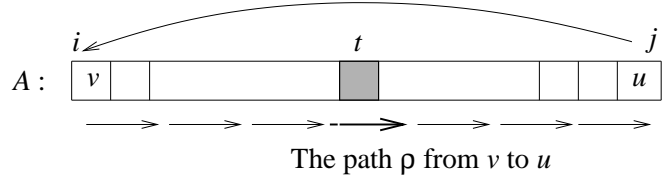}}
\caption{There has to be an edge in $\rho$ from a descendant of $v$ in
$A[i..t]$ to an ancestor of $u$ in $A[t..j]$ or $A[t] \in \anc\cap\des$.}
\label{figure3}
\end{figure}

Hence either there is an element in $\des\cap\anc$, the only
such element can be $A[t]$, or 
there has to be an edge from a vertex in $\des$ to a vertex in $\anc$,
where $\des$ is the set of descendants
of $v$ in $A[i..t]$ and $\anc$ is the set of ancestors of $u$
in $A[t..j]$. 
\end{proof}

Lemma~\ref{cycle-check} shows us that cycle detection is easy, so
let us assume henceforth that the edge $(u,v)$ does not create
a cycle, hence the graph $G$ is still a DAG. The following lemma 
shows that the number $t$ has the desired property that we were looking for.

\begin{lemma}
\label{correct-t}
The number $t$ has the property that the number of descendants 
of $v$ in $A[i..t]$ is equal to the number of ancestors of $u$ in
$A[(t+1)..j]$.
\end{lemma}
\begin{proof}
We have  $t = \rightptr = \leftptr$.  When we terminate the while loop,
$\anc$ is the set of
ancestors of $u$ in $A[t..j]$ and 
$\des$ is the set of
descendants of $v$ from $A[t..j]$. Since $(u,v)$ does not create a cycle,
$A[t]$ is either a descendant of $v$ or an ancestor of $u$, but not both.

If the main while loop got broken because of the first break statement, then 
$A[t]$ is a descendant of
$v$ and if the main while loop got broken due to the second break statement, 
then $A[t]$ is an ancestor of $u$.
Since we interleave adding a vertex to the set
$\anc$ with adding a vertex to the set $\des$, if we break because of the 
first break statement, we have $|\anc| = |\des| $; and if we
break because of the second break statement, we have $|\anc| = |\des|+1$.
Thus in both cases, it holds that the number of descendants of $v$
in $A[i..t]$ is equal to the number of ancestors of $u$ in $A[(t+1)..j]$.
\end{proof}

At the end of $\mathsf{Phase}~1$, since $G$ is still a DAG, 
we delete all the elements of $\anc\cup\des$ 
from their current locations in $A$.
The vertices in $\anc$ have to find new places in $A[i..t]$ and the vertices
in $\des$ have to find new places in $A[(t+1)..j]$.\footnote{Note that
the number of empty locations in $A[i..t]$ exactly equals $|\anc|$ and
the number of empty locations in $A[(t+1)..j]$ exactly equals $|\des|$.}
  This forms $\mathsf{Phase}~2$ of our algorithm.

%\paragraph{$\mathsf{Phase}~2$:}
%\begin{itemize}
%\item Traverse the subarray $A[t..i]$ (from $t$ to $i$ backwards) 
%and delete/insert elements in $\anc$.
% Elements now get deleted from $\anc$ and 
%some new elements can also get get inserted into $\anc$.
%Whenever we see an empty location in $A$ 
%(created either by deleting 
%an element of $\des$ from there
%or because of inserting a new element from there to $\anc$), 
%the head of the queue $\anc$ is deleted from $\anc$ and is assigned to 
%that empty location.
%Ancestors of elements of $\anc$ in these locations get deleted from
%$A$ (creating an empty location in $A$) and inserted to $\anc$.
%\item A symmetric step takes place in the subarray $A[(t+1)..j]$ with
%respect to the queue $\des$.
%\end{itemize}

%\paragraph
\medskip

{\bf \textsc{Phase}~2.}
We now describe $\mathsf{Phase}~2$ from the side of $\anc$ as
Algorithm~\ref{fig-algo2}.
(A symmetric procedure also takes place on the side of $\des$.)
In this phase vertices get deleted from $\anc$
and vertices can also get added to $\anc$. 
\begin{algorithm}[h]
\begin{algorithmic}
\STATE{$\leftptr = t$;}
\WHILE{$\leftptr \ge i$}
\IF{$A[\leftptr]$ is an empty location} 
\STATE{-- delete the head of the queue $\anc$, call it $h$, and
set $A[\leftptr] = h$}
\ELSIF{$A[\leftptr]$ (call it $r$) is a {\em predecessor} of some element in 
$\anc$}
\STATE{-- insert $r$ to the queue $\anc$}
\STATE{ -- delete the  head of the queue $\anc$, call it $h$, and set $A[\leftptr] = h$}
\ENDIF
\STATE{$\leftptr = \leftptr - 1$;}
\ENDWHILE
\end{algorithmic}
\caption{$\mathsf{Phase}~2$ of our algorithm from the side of $\anc$.}
\label{fig-algo2}
\end{algorithm}

In $\mathsf{Phase}~2$, the pointer $\leftptr$ traverses the 
subarray $A[t..i]$ (from $t$ leftwards to $i$)
and elements get deleted/inserted in $\anc$.
% Elements now get deleted from $\anc$ and 
%some new elements can also get get inserted into $\anc$.
Whenever $\leftptr$ sees an empty location in $A$,
the head of the queue $\anc$ is deleted from $\anc$ and is assigned to 
that empty location. Whenever $A[\leftptr]$ is a predecessor of some
element in $\anc$, then $A[\leftptr]$ is removed from that location and 
is inserted into $\anc$ and the current head of the queue $\anc$ is
inserted into that location.
%Ancestors of elements of $\anc$ in these locations get deleted from
%$A$ (creating an empty location in $A$) and inserted to $\anc$.
We have the following lemma, which is simple to show. Its proof is
included in the Appendix.
\begin{lemma}
\label{lemma3}
The subroutine in Algorithm~\ref{fig-algo2} 
maintains the following invariant in every iteration
of the while loop: the number of elements in $\anc$ equals the number of
empty locations in the subarray $A[i..\leftptr]$.
\end{lemma}
\begin{proof}
It is easy to see that the invariant is true at the beginning of the
subroutine. In other words, at the end of $\mathsf{Phase}~1$, the 
number of empty locations in $A[i..t]$ exactly equals $|\anc|$.

We will now show that this invariant is maintained throughout 
$\mathsf{Phase}~2$.
%Say $|\anc| = k$ and there are $k$ empty locations in the subarray
%$A[i..\leftptr]$.
Whenever $\leftptr$  sees an empty location in $A$,  we
delete an element from $\anc$, hence this invariant is maintained since the 
number of empty locations in $A[i..\leftptr]$ 
decreases by one and so does the size of
$\anc$. Whenever $\leftptr$ sees a predecessor $p$ 
of some element of $\anc$ at the current location, then $p$ is deleted
from its current location $\leftptr = \ell$ in $A$ and $p$
is inserted into $\anc$; the leading element $h$ of $\anc$ gets deleted from 
$\anc$ and we assign $A[\ell] = h$. Thus neither the number of
empty locations in $A[i..\leftptr]$ nor the size of $\anc$ changes by
our deletion and insertion, so the
invariant is maintained. Hence when we exit the while loop, which is
when $\leftptr = i-1$, the queue $\anc$ will be empty.
\end{proof}

We then  traverse the subarray $A[(t+1)..j]$ from location $t+1$ to
location $j$ and execute the algorithm in Algorithm~\ref{fig-algo2} 
with respect to $\des$. For the sake of completeness, we present that 
algorithm as Algorithm~\ref{fig-algo3} below.

\begin{algorithm}[h]
\begin{algorithmic}
\STATE{$\rightptr = t+1$;}
\WHILE{$\rightptr \le j$}
\IF{$A[\rightptr]$ is an empty location} 
\STATE{-- delete the head of the queue $\des$, call it $\ell$, and
set $A[\rightptr] = \ell$}
\ELSIF{$A[\rightptr]$ (call it $r$) is a {\em successor} of some element in 
$\des$}
\STATE{-- insert $r$ to the queue $\des$}
\STATE{ -- delete the  head of the queue $\des$, call it $\ell$, and set $A[\rightptr] = \ell$}
\ENDIF
\STATE{$\rightptr = \rightptr + 1$;}
\ENDWHILE
\end{algorithmic}
\caption{$\mathsf{Phase}~2$ of our algorithm from the side of $\des$.}
\label{fig-algo3}
\end{algorithm}

A lemma analogous to Lemma~\ref{lemma3} will show that there is always enough
room in the array $A[\rightptr..j]$ to accommodate the elements of $\des$
(refer Algorithm~\ref{fig-algo3}).
This completes the description of our algorithm to update the topological
ordering when a new edge is inserted.

\subsubsection{Correctness.}
%\label{correct}
We would now like to claim that after running $\mathsf{Phase}~1$ and
$\mathsf{Phase}~2$ of our algorithm, we have a valid topological ordering.
Our ordering is defined in terms of the array $A$. Each element $x$ that
has been assigned a new location in $A$ has a new $\ord$ value, which is 
the index of its
new location. For elements that never belonged to $\anc\cup\des$, 
the $\ord$ value is unchanged. For the
sake of clarity, let us call the ordering before the new edge $(u,v)$ was
inserted as $\ord$ and let ${\ord}'$ denote the new function after executing
our algorithm. We will show the following theorem here
(its proof is included in the Appendix).
\begin{theorem}
\label{main-correct}
The function ${\ord}'$ is a valid topological ordering.
\end{theorem}
\begin{proof}
We need to show that $\ord'$ is a valid topological ordering.
Consider any edge $(x,y)$ in the graph. 
We will show that ${\ord}'(x) < {\ord}'(y)$. We will split this into
three cases.
\begin{itemize}
\item {$x \in \anc$.} There are three further cases: (i) $y \in \anc$, 
(ii) $y \in \des$, (iii) $y$ is neither in $\anc$ nor in $\des$. 
In case (i), both $x$ and $y$ are in $\anc$ and since there is an edge
from $x$ to $y$, the vertex $y$ is ahead of $x$ in the queue $\anc$. So
$y$ gets deleted from $\anc$ before $x$ and is hence assigned a higher
indexed location in $A$ than $x$. In other words, ${\ord}'(x) < {\ord}'(y)$.
In case~(ii), we have ${\ord}'(x) \leq t < {\ord}'(y)$. 
In case~(iii), since
elements of $\anc$ move to lower indexed locations in $A$,  we have $\ord'(x) < \ord(x)$
whereas $\ord'(y) = \ord(y)$; since ${\ord}(x) < {\ord}(y)$, we get ${\ord}'(x) < {\ord}'(y)$.

\vspace*{0.2cm}

\item {$x \in \des$.} There are only two cases here: (i) $y \in \des$ or
(ii) $y$ is neither in $\anc$ nor in $\des$. This is because
if $x\in\des$ and $y\in\anc$, then $y \leadsto u \rightarrow v \leadsto x$. 
So $(x,y) \in E$ creates a cycle. This is impossible
since we assumed that after inserting the edge $(u,v)$, $G$ remains a DAG. 
Thus we cannot have $x \in \des$ and $y\in\anc$ for $(x,y)\in E$.

In case (i) here, because $(x,y) \in E$, the vertex $x$ is ahead of the vertex $y$
in the queue $\des$, so
$x$ gets deleted from $\des$ before $y$ and is hence assigned a lower
indexed location in $A$ than $y$. Equivalently, ${\ord}'(x) < {\ord}'(y)$.
In case (ii) here, since $x \in \des$ and $y \notin \des$, it has to be that
either $\ord(y) > \ord(u)$ in which case ${\ord}'(x) < {\ord}(y) = \ord'(y)$, 
or by the time $\rightptr$ visits the
location in $A$ containing $y$, the vertex $x$ is already deleted from the
queue $\des$ - otherwise, $y$ would have been inserted into $\des$. Thus 
${\ord}'(x) < {\ord}(y) = {\ord}'(y)$.

\vspace*{0.2cm}

\item {$x \notin \anc\cup\des$.} There are three cases again here: 
(i) $y \in \anc$, 
(ii) $y \in \des$, (iii) $y$ is neither in $\anc$ nor in $\des$.
The arguments here are similar to the earlier arguments and it is easy to 
check that in all three cases we have ${\ord}'(x) < {\ord}'(y)$.
\end{itemize}
Thus $\ord'$ is a valid topological ordering.
\end{proof}

Thus our algorithm is correct. In Section~\ref{runtime} we will show that
its running time, summed over all edge insertions, is $O(n^{5/2})$.

\section{Running Time Analysis}
\label{runtime}

The main tasks in our algorithm to update $\ord$ to $\ord'$ 
(refer to Algorithms~\ref{fig-algo}, \ref{fig-algo2}, and \ref{fig-algo3}) 
are: \\
%\begin{enumerate}
(1) moving a pointer $\leftptr$ from
location $j$ to $i$ in the array $A$ and checking if $A[\leftptr]$ is
a {\em predecessor} of any element of $\anc$;\\
(2) moving  a pointer $\rightptr$ from
location $i$ to $j$ in the array $A$ and checking if $A[\rightptr]$ is
a {\em successor} of any element of $\des$;\\
(3) checking at the end of $\mathsf{Phase}~1$, if $(x,y) \in E$
for any $x$ in 
$\{x: i \le \ord(x) \le t \ \mathsf{and} \ v \leadsto x\}$ and 
any $y$ in 
$\{y: t \le \ord(y) \le j \ \mathsf{and} \ y \leadsto u\}$.
(If so, then $(u,v)$ creates a cycle.)
%\end{enumerate}

Lemma~\ref{check-G-dag} bounds the cost taken by Step~3 over all iterations.
It can be proved using a potential function argument.
% and its proof is given in the Appendix.

\begin{lemma}
\label{check-G-dag}
The cost for task~3, stated above, summed over
{\em all} edge insertions, is $O(n^2)$.
\end{lemma}
\begin{proof}
We need to check if there is an edge $(x,y)$ between  
some $x \in \des$ and some $y \in \anc$. 
We pay a cost of $|\anc|\cdot|\des|$  
for checking $|\anc|\cdot|\des|$ many entries of the matrix $M$. 

Case(i): There is no adjacent pair $(x,y)$ for $x \in \des, y\in \anc$.
Then the cost $|\des|\cdot|\anc|$ can be bounded by
$N(e)$, which is the number of pairs of vertices $(y,x)$ for which the 
relationship $y \leadsto x$ 
has started now for the first time (due to the insertion of the edge $(u,v)$).
Recall that at the end of $\mathsf{Phase}~1$, $\des$ is the set 
$\{x: i \le \ord(x) \le t \ \mathsf{and} \ v \leadsto x\}$ and
$\anc$ is the set 
$\{y: i \le \ord(y) \le t \ \mathsf{and} \ y \leadsto u\}$. Thus
each vertex in this set $\des$ currently has a lower $\ord$ value than
each vertex in $\anc$ - so the only relationship that could have 
existed between such an $x$ and $y$  is $x \leadsto y$, which we have
ensured does not exist. Thus these pairs $(x,y)$ were incomparable and now 
the relationship $y \leadsto x$ has been established.
It is easy to see that
$\sum_{e\in E} N(e)$ is at most $n\choose 2$ since any pair of vertices
can contribute at most 1 to $\sum_{e\in E} N(e)$.

Case(ii): There is indeed an adjacent pair $(x,y)$ for 
$x \in \des, y\in \anc$. Then
we quit, since $G$ is no longer a DAG. The check that showed $G$ to contain
a cycle cost us $|\anc|\cdot|\des|$, which is $O(n^2)$. We pay this cost only 
once as this is the last step of the algorithm. 
\end{proof}

Let $\ord_e$ be our valid topological
ordering before inserting edge $e$ and let $\ord'_e$ be our valid topological
ordering after inserting $e$.  Lemma~\ref{cost-steps-1and2} is our first step in bounding 
the cost taken
for tasks 1 and 2 stated above.
\begin{lemma}
\label{cost-steps-1and2}
The cost taken for tasks 1 and 2, stated above, is $\sum_{x\in V}|\ord_e(x) - \ord'_e(x)|$.
\end{lemma}
\begin{proof}
In Step~1 we find out if $A[\leftptr] = x$ is
a predecessor of any element of $\anc$ by checking the entries
$M[x,w]$ for each $w \in \anc$. Each element $w$ which is currently in
$\anc$ pays unit cost for checking the entry $M[x,w]$.

Any element $w\in\anc$ belongs to the set $\anc$ while
$\leftptr$ moves from location $\ord_e(w)$ to
$\ord'_e(w)$. When $\leftptr = \ord_e(w)$ and we identify $A[\ord_e(w)]=w$ to be 
a predecessor of some element in $\anc$, the vertex $w$ gets inserted into
$\anc$. When $\leftptr$ is at some empty location $\beta$ and 
the vertex $w$ is the head of the queue $\anc$, then $w$ is deleted from $\anc$ 
and we set $A[\beta] = w$, which implies that $\ord'_e(w) = \beta$.
So the total cost paid by $w$ is  $\ord_e(w)-\ord'_e(w)$, which is
to check the entries $M[A[\leftptr],w]$
as the pointer $\leftptr$ moves from location $\ord_e(w)-1$ to $\ord'_e(w)$.

Symmetrically, for any
vertex $y$ that belonged to $\des$ during the course of the
algorithm, the cost paid by $y$ is $\ord'_e(y)-\ord_e(y)$. A vertex $z$
that never belonged to $\anc \cup \des$, does not pay anything and 
we have $\ord'_e(z) = \ord_e(z)$. Thus for each
$x \in V$, the cost paid by $x$ to move the pointers $\leftptr$/$\rightptr$
is $|\ord_e(x) - \ord'_e(x)|$. 
\end{proof}

We will show the
following result in Section~\ref{main-lemma}.
\begin{lemma}
\label{lemma4}
$\sum_{e\in E}\sum_{x \in V}|\ord_e(x)-\ord'_e(x)|$ is $O(n^{5/2})$, where
$\ord_e$ is our valid topological
ordering before inserting edge $e$ and $\ord'_e$ is our valid topological
ordering after inserting $e$.
\end{lemma}
Theorem~\ref{thm1}, stated in Section~\ref{intro}, follows from
 Theorem~\ref{main-correct}, \ Lemmas~\ref{check-G-dag} and \ref{lemma4}. 
Also note that the space requirement of our algorithm is $O(n^2)$, 
since our algorithm uses only the $n\times n$ adjacency matrix $M$, 
the array $A$,
the queues $\anc,\des$, and the pointers $\leftptr,\rightptr$.

\subsection{Proof of Lemma~\ref{lemma4}}
\label{main-lemma}
Let $e=(u,v)$ and let $\ord_e$ be our topological
ordering before inserting $(u,v)$ and $\ord'_e$ our topological
ordering after inserting $(u,v)$.
Our algorithm for updating $\ord_e$ to $\ord'_e$ 
basically permutes the vertices in 
the subarray $A[i..j]$.
The elements which get inserted into the queue $\anc$ move to lower
indexed locations in $A$ (compared to their locations in $A$ before $e$ was added),
elements which get inserted into the queue $\des$ move to higher
indexed locations 
%(compared to their locations in $A$ before $e$ was added)
in $A$, and elements which did not get inserted into either $\anc$ or $\des$
remain unmoved in $A$. So our algorithm is essentially a permutation $\pi_e$ of
elements that belonged to $\anc \cup \des$.

Let ${\anc}_e$ denote the ordered set of all vertices that 
get inserted to the set $\anc$ in Algorithms~\ref{fig-algo} 
and \ref{fig-algo2} 
(and of course, later get deleted from $\anc$ in Algorithm~\ref{fig-algo2})
while inserting the edge $e$.
In other words, these are the vertices $w$
for which $\ord_e(w) > \ord'_e(w)$.
Define $\des_e$ as the ordered set of all those vertices $w$ for which 
$\ord_e(w) < \ord'_e(w)$. Equivalently, these are all the vertices that get
inserted into $\des$ in Algorithms~\ref{fig-algo} and \ref{fig-algo3}.
Let $\anc_e = \{u_0,u_1,\ldots,u_k\}$, where $u_0 = u$ and 
$\ord(u_0) > \ord(u_1) > \cdots > \ord(u_k)$, and
let
$\des_e = \{v_0,v_1,\ldots,v_s\}$, where $v_0 = v$, and
$\ord(v_0) < \ord(v_1) < \cdots < \ord(v_s)$.

Let us assume that all the vertices of $\anc_e\cup\des_e$  are in their 
old locations in $A$ (those locations given by the ordering $\ord_e$;
so $A[i] = v$ and
$A[j] = u$). We will now decompose the permutation $\pi_e$ on these elements
into a composition of swaps. 
Note that our algorithm does not perform any swaps, however to prove
Lemma~\ref{lemma4}, it is useful to view $\pi_e$ as a composition of 
appropriate swaps. 
The function $\swap(x,y)$ takes as input: $x \in \anc_e$ and $y \in \des_e$,
where $\ord(x) > \ord(y)$, 
and swaps the occurrences of $x$ and $y$ in the array $A$. That is, 
if $A[h] = x$ and $A[\ell] = y$, where
$h > \ell$, then $\swap(x,y)$ makes $A[\ell] = x$ and $A[h] = y$.

The intuition behind decomposing $\pi_e$ into swaps
between such an element $x \in \anc_e$ and such an element $y \in \des_e$
is that we will have the following
useful property: {\em whenever we
swap two elements $x$ and $y$, it is always the case that
$\ord(x) > \ord(y)$ and 
we will never swap $x$ and $y$ again in the future
(while inserting other new edges)
since we now have the relationship $x\leadsto y$, so $\ord(y) > \ord(x)$
has to hold from now on.}
%We have the useful property that whenever we
%swap two elements $x$ and $y$, %it is always the case that
% $\ord(x) > \ord(y)$ and 
%we will never swap $x$ and $y$ again in the future
%(while inserting other new edges)
%since we now have the relationship $x\leadsto y$, so $\ord(y) > \ord(x)$
%has to hold from now on.
%We will use the symbol $\ord_e$ (we used $\ord$ so far for this function) 
%to represent our valid topological ordering
%just before $e$ is inserted,
%Once all the swaps are 
%performed, we will show that $\ord$ will denote the function $\ord'_e$.

We will use the symbol 
$\overline{\ord}$ to indicate the {\em dynamic} inverse function of $A$, 
so that $\overline{\ord}$ reflects instantly 
changes made in the array $A$.
So as soon as we swap $x$ and $y$ so that $A[\ell] = x$ and $A[h] = y$,
we will say $\overline{\ord}(x) = \ell$ and $\overline{\ord}(y) = h$. 
Thus the function $\overline{\ord}$ gets initialized to the function
$\ord_e$, it gets updated with every swap that we perform and finally becomes
the function $\ord'_e$.
%Currently, the function
%$\ord = \ord_e$ (that is, before we move the elements of $\anc_e\cup\des_e$).
%We present below a simple method
%for decomposing the permutation $\pi_e$ into appropriate swaps.  
\subsubsection{3.1.1 Decomposing $\pi_e$ into appropriate swaps.}
\begin{algorithmic}
\STATE{-- Initialize the permutation $\pi_e$ to identity and the function $\overline{\ord}$ to $\ord_e$.}
\FOR{ $x \in \{u_k,u_{k-1},\ldots,u_0\}$  (this is $\anc_e$: elements in reverse order of insertion into $\anc_e$)} 
%we have $\overline{\ord}(u_k) < \cdots < \overline{\ord}(u_0)$)}
\FOR{ $y \in \{v_s,v_{s-1},\ldots,v_0\}$  (this is $\des_e$: elements in reverse order of insertion into $\des_e$) }
%we have $\overline{\ord}(v_s) > \cdots > \overline{\ord}(v_0)$)}
\IF{$\overline{\ord}(x) > \overline{\ord}(y)$}
\STATE{$\pi_e = \swap(x,y) \circ \pi_e$ \hspace*{0.3in} ($\ast$)}\hspace*{0.2in}
\COMMENT{Note that swapping $x$ and $y$ changes their $\overline{\ord}$ values.}
\ENDIF
\ENDFOR
\ENDFOR
\STATE{-- Return $\pi_e$ (as a composition of appropriate swaps).}
\end{algorithmic}

To get a better insight into this decomposition of $\pi_e$, 
let us take the example
of the element $u_k \in \anc$
($u_k$ has the minimum $\ord_e$ value among all the elements in $\anc$).
%and see if the $\swap(u_k,\cdot)$'s included in $\pi_e$ 
%indeed capture how $\pi_e$ acts on $u_k$.
Let $v_0,v_1,\ldots,v_r$ 
be the elements of $\des_e$ whose $\ord_e$ value is
less than $\ord_e(u_{k}) = \alpha$. Recall that 
$\ord_e(v_0) < \cdots < \ord_e(v_r) < \ord_e(v_{r+1}) < \cdots \ord_e(v_s)$.
When the outer {\bf for} loop for $x = u_k$
is executed, $u_k$ does not swap with $v_s,\ldots,v_{r+1}$. The first
element that $u_k$ swaps with is $v_r$, then it swaps with $v_{r-1}$,
so on, and $u_k$ finally swaps with $v_0$ and takes the location $i$ in $A$
that was occupied by $v_0$.
Thus ${\ord}'_e(u_k) = i$ and 
$\overline{\ord}(v_0),\ldots,\overline{\ord}(v_r)$ are higher than what they were formerly, since
each $v_{\ell} \in \{v_0,\ldots,v_{r-1}\}$ is currently occupying the
location that was formerly occupied by $v_{\ell+1}$, and $v_r$ is occupying
$u_k$'s old location $\alpha$. Thus the total movement of $u_k$ from
location $\alpha$ to location $i$, can be written as:
\[ \ord_e(u_k) - \ord'_e(u_k) = \alpha - i = \sum_{v_{\ell} \in \{v_r,\ldots,v_0\}}d(u_k,v_{\ell})\]
where $d(u_k,v_{\ell}) = \overline{\ord}(u_k) - \overline{\ord}(v_{\ell})$
%$u_k$ and $v_{\ell}$ 
%when we swap $u_k$ and $u_{\ell}$. %that is, 
when $\swap(u_k,v_{\ell})$ is included in $\pi_e$ (refer to ($\ast$) in Section 3.1.1).

\paragraph{Correctness of our decomposition of $\pi_e$.} It is easy to see 
that the composition of swaps, $\pi_e$, that is returned by the above
method in Section 3.1.1, when applied on 
$\anc_e=\{u_k,u_{k-1},\ldots,u_0\}$ and 
$\des_e=\{v_0,\ldots,v_{s-1},v_s\}$, results in these elements occurring in 
the relative order:
$u_k,u_{k-1},\ldots,$ $u_0,v_0,\ldots,v_{s-1},v_s$ in $A$.
We claim that our algorithm 
(Algorithms~\ref{fig-algo},~\ref{fig-algo2},~\ref{fig-algo3}) places these 
elements in the same order in $A$. This is because
our algorithm maintained both $\anc$ and $\des$ as queues - so elements
of $\anc_e$ (similarly, $\des_e$) do not cross each other, so $u_k,\ldots,u_0$
(resp., $v_0,\ldots,v_s$) will be placed in this order, 
from left to right, in $A$. Also, we insert
all elements of $\anc_e$ in $A[i..t]$
and all elements of $\des_e$ in $A[(t+1)..j]$, thus our algorithm puts
elements of $\anc_e\cup\des_e$ in the order $u_k,u_{k-1},\ldots,$ $u_0,v_0,\ldots,v_{s-1},v_s$ in $A$.
Thus we have
obtained a correct decomposition (into swaps) 
of the permutation performed by our algorithm.

For every pair $(x,y) \in \anc_e\times\des_e$, if $\swap(x,y)$ is included
in $\pi_e$ (see ($\ast$)), define 
$d(x,y) =  \overline{\ord}(x) - \overline{\ord}(y)$, 
the difference in the location indices occupied by $x$
and $y$, when $\swap(x,y)$ gets included in $\pi_e$.
For instance, $d(u_k,v_0) = \ord_e(v_1)-\ord_e(v_0)$ since $u_k$
is in the location $\ord_e(v_1)$ (due to swaps with $v_r,\ldots,v_1$) and
$v_0$ is unmoved in its original location $\ord_e(v_0)$, at the
instant when $\swap(u_k,v_0)$ gets included in $\pi_e$.
%In other words,
%for each pair $(x,y)$ such that $\swap(x,y)$ occurs in $\pi_e$, let
%$ d(x,y) = \ord(x)-\ord(y)$, where $\ord(\cdot)$ captures changes made 
%to $\ord_e$ 
%by swaps included in $\pi_e$ till the instant $\swap(x,y)$ is included
%in $\pi_e$ (in ($\ast$)).

Since we broke the total movement in $A$
of any $x \in \anc_e$ (which is $\ord_e(x)-\ord'_e(x)$) into
a sequence of swaps with certain elements in $\des_e$, we have
for any $x \in \anc_e$
\[\ord_e(x) - \ord'_e(x) = \sum_{y:(x,y) \in \pi_e}d(x,y),\]
where we are using ``$(x,y) \in \pi_e$'' to stand for ``$\swap(x,y)$ exists in 
$\pi_e$''.
%Now we are ready to prove Lemma~\ref{lemma4}.
%We need to show that $\sum_{e\in E}\sum_{x\in V}|\ord(x) - \ord'(x)|$ is $O(n^{5/2})$.
%The work done to insert all edges in $e \in E$ can be regarded as
%applying the permutations $\pi_e$ one after the other. 
%If we consider each permutation $\pi_e$ as a set of swaps, then
%note that the term $\swap(x,y)$
%for any pair $(x,y)$ can occur at most once in all the permutations.
%We saw that $\sum_{x\in V}\Delta_e(x) = \sum_{y\in\des_e}d(x,y)$.
\begin{eqnarray}
\text{We have} \ \sum_{x \in V}|\ord_e(x)-\ord'_e(x)| & = & \sum_{w\in{\anc}_e}(\ord_e(w)-\ord'_e(w)) + \sum_{y\in{\des}_e}(\ord'_e(y)-\ord_e(y)) \\
& = & 2\sum_{x \in {\anc}_e}(\ord_e(x) - \ord'_e(x))\\
& = & 2\sum_{x \in \anc_e} \ \sum_{y:(x,y)\in\pi_e}d(x,y)\\ %= 2\sum_{(x,y):(x,y)\in\pi_e}d(x,y).
& = & 2\sum_{(x,y):(x,y)\in\pi_e}d(x,y).
\end{eqnarray}
Equality~(2) follows from (1) because 
%$\sum_{w \in {\anc}_e}(\ord(w) - \ord'(w)) = 
%\sum_{y \in {\des}_e}(\ord'(y) - \ord(y))$ 
$\sum_{x\in V}\ord(x) = 
\sum_{x\in V}\ord'(x)$. Equality~(3) follows from the preceding paragraph.
%Equality~(3) follows from (2) and the prece
So the entire running time to insert all edges in $E$ is
$2\sum_{e\in E}\sum_{(x,y):(x,y)\in\pi_e}d(x,y)$. 

We now claim that for any pair $(x,y)$, we can have $(x,y) \in \pi_e$ 
for at most one permutation $\pi_e$.
For $\swap(x,y)$ to exist in $\pi_e$, we need (i)~$(x,y)\in\anc_e\times\des_e$,
and (ii)~$\ord(x) > \ord(y)$.
Once $\pi_e$ swaps $x$ and $y$, 
subsequently $x \leadsto y$ (since $x \leadsto u \rightarrow v \leadsto y$) 
and $\ord(y) > \ord(x)$,
so $(x,y)$ can never again satisfy $\ord(x) > \ord(y)$. 
So for any pair $(x,y)$, $\swap(x,y)$ can occur in at most
one permutation $\pi_e$ over all $e\in E$. 
%That is, $(x,y) \in \pi_e$ for at most one $e \in E$.
%Since any pair $(x,y)$ can be swapped at most once over all
%the permutations $\pi_e, e\in E$,
Thus we have:
\begin{equation}
\label{ineq1}
\sum_{e\in E} \sum_{(x,y): (x,y)\in\pi_e}d(x,y)
 = \sum_{(x,y):(x,y)\in\pi_e \ \text{for some $e$}} d(x,y).
\end{equation}
Note that the summation on the right hand side in Inequality~(\ref{ineq1}) 
is over all those pairs $(x,y) \in V \times V$ 
such that $\swap(x,y)$ exists in some $\pi_e$, for $e \in E$.
%We can in fact make a stronger statement that the left hand side and the right
%hand side of Inequality~(\ref{ineq1}) are equal since the LHS obviously
%counts every pair that appears in the RHS. However, what we need is only the
%above inequality in order to prove Lemma~\ref{lemma4}.

The following lemma 
was shown in \cite{AFM06}\footnote{The algorithm in \cite{AFM06} 
performs swaps to obtain a valid topological ordering 
and Lemma~\ref{afm-lemma} is used in their analysis to show an
$O(n^{2.75})$ upper bound for the running time of their algorithm.}.  
This finishes the proof of Lemma~\ref{lemma4}.

\begin{lemma}
\label{afm-lemma}
$\sum d(x,y)$ is $O(n^{5/2})$, where the summation is over all those pairs $(x,y)$
such that $\swap(x,y)$ exists in some permutation $\pi_e, e \in E$.
%such that there is some permutation $\pi_e$, $e\in E$  that swaps $x$ and $y$.
\end{lemma}
\begin{proof}We present the proof of this lemma given in \cite{AFM06}.
We need to show that $\sum_{x,y}d(x,y)$ is $O(n^{5/2})$. 
Let $\ord^*$ denote the final topological ordering. Define
\[ X(\ord^*(x),\ord^*(y)) = \begin{cases} d(x,y)   & \text{if there is some permutation $\pi_e$ that swaps $x$ and $y$}\\
				0 & \text{otherwise.}
\end{cases}\]
Since $\swap(x,y)$ can occur in at most one permutation $\pi_e$, the variable 
$X(i,j)$ is clearly 
defined. Next, we model a few linear constraints on $X(i,j)$, formulate it as 
a linear program and use this LP to prove that 
$\max\{\sum_{i,j} X(i,j)\} = O(n^{5/2})$. By definition of $d(x,y)$ and 
$X(i,j)$, 
\[0\leq X(i,j) \leq n, \mbox{ for all } i,j \in \{1\ldots n\}.\]
For $j \leq i$, the corresponding edges $(\ord^{* -1}(i),\ord^{* -1}(j))$ go 
backwards and thus are never inserted at all. Consequently, 
$$X(i,j) = 0 \mbox{ for all } j \leq i.$$ 
Now consider an arbitrary vertex $w$, which is finally at position $i$, i.e., 
$\ord^*(w) = i$. Over the insertion of all the edges, this vertex has been moved 
left and right via swapping with several other vertices. Strictly speaking, 
it has been swapped left with vertices at final positions $j>i$ and has been 
swapped right with vertices at final position $j<i$. Hence, the overall movement 
to the left is $\sum_{j>i}X(i,j)$ and to the right is $\sum_{j<i}X(j,i)$. 
Since 
the net movement (difference between the final and the initial position) must 
be less than $n$, 
\[ \sum_{j>i}X(i,j) - \sum_{j<i}X(j,i) \leq n \mbox{ for all } 1\leq i \leq n.\]
Putting all the constraints together,
we aim to solve the following linear program.
$$\max \sum_{1\leq i \leq n \mbox{, } 1\leq j \leq n} X(i,j) \mbox{ such that }$$
\begin{enumerate}
\item[(i)] $X(i,j) = 0$ for all $1\leq i \leq n$ and $1 \leq j \leq i$
\item[(ii)] $0 \leq X(i,j) \leq n$ for all $1 \leq i \leq n$ and $i<j\leq n$
\item[(iii)] $\sum_{j>i}X(i,j)-\sum_{j<i}X(j,i) \leq n-1, \mbox{ for all } 1\leq i \leq n$
\end{enumerate}
%Note that these are necessary constraints, but not sufficient. But this is 
%enough for our purpose as an upper bound to the solution of this LP will give 
%an upper bound for $\sum X(i,j)$ in our analysis. 
In order to prove the 
upper bound on the solution to this LP, we consider the dual problem:
\[\min \left[ n\sum_{0 \leq i \leq n\mbox{, }i<j<n} Y_{i \cdot n + j} + n\sum_{0 \leq i < n} Y_{n^2 + i} \right] \mbox{ such that}\]
\begin{enumerate}
\item[(i)] $Y_{i \cdot n + j} \geq 1$ for all $0 \leq i < n$ 
and for all $j \leq i$
\item[(ii)]$Y_{i \cdot n + j} + Y_{n^2 + i} - Y_{n^2 + j} \geq 1$ 
for all $0 \leq i<n$ and for all $j>i$ 
\item[(iii)]$Y_i \geq 0$ for all $0 \leq i<n^2 + n$ 
\end{enumerate}
and the following feasible solution for the dual:
\begin{eqnarray*}
Y_{i \cdot n + j} & = & 1 \mbox{ for all } 0\leq i<n \mbox{ and for all } 0 \leq j \leq i
\\Y_{i \cdot n + j} &=& 1 \mbox{ for all } 0\leq i<n \mbox{ and for all } i<j \leq i+1+2\sqrt{n}
\\Y_{i \cdot n + j} &=& 0 \mbox{ for all } 0\leq i<n \mbox{ and for all } j>i+1+2\sqrt{n}
\\Y_{n^2 + i} &=& \sqrt{n-i} \mbox{ for all } 0\leq i < n.
\end{eqnarray*}
The solution has a value of $n^2 + 2n^{5/2} + n\sum_{i=1}^{n} \sqrt{i} = O(n^{5/2})$, 
which by the primal-dual theorem is a bound on the solution of the original LP. 
This completes the proof of Lemma~\ref{afm-lemma}
and thus Lemma~\ref{lemma4} is proved. 
\end{proof}

\section{The $O((m+n\log n)\sqrt{m})$ algorithm}
\label{new-algo}
In this section we present an algorithm with running time $O((m + n\log n)\sqrt{m})$ 
for online topological ordering.  This algorithm is an adaptation of the
algorithm by Katriel and Bodlaender in \cite{KB05} and uses the 
{\em Ordered List} data structure from \cite{DS87}, also used in \cite{KB05} for this problem.  
That is, the function $\ord$ on $V$ is maintained by a data structure $ORD$
which is a data structure that allows  a total order to be maintained over a list of items.
Each item $x$ in $ORD$ has an associated integer label $\ord(x)$ and 
the label associated with $x$ is smaller than the label associated with
$y$, iff $x$ precedes $y$ in the total order.
The following operations can be performed in constant amortized time 
[see Dietz and Sleator \cite{DS87}, Bender~et~al. \cite{BCDFZ02} for details]: 
the query {\em Order}$(x,y)$ determines whether $x$ precedes $y$ or $y$ precedes
$x$ in the total order (i.e., if $\ord(x) < \ord(y)$ or 
$\ord(y) < \ord(x)$),
{\em InsertAfter}$(x,y)$ ({\em InsertBefore}$(x,y)$) inserts the item $x$ 
immediately after (before) the item $y$ in the total order, and
{\em Delete}$(x)$ removes the item $x$.
%the query {\em Order}$(x,y)$ determines whether $x$ precedes $y$ or $y$ 
%precedes $x$ in the total order, and 
%{\em Next}$(x)$ (similarly, {\em Prev}$(x)$) returns the item that 
%appears immediately after (resp. before) $x$ in the total order.

When a new edge ($u,v$) is added to a graph $G$, there are two cases: 
(i) either $Order(u,v)$ is true, in which case the current ordering of
elements in $ORD$ is still a valid ordering, so we need to do nothing
except add $(u,v)$ in the list of edges incoming into $u$ and in the
list of edges going out of $v$;
(ii)~$Order(u,v)$ is false, in which case the edge $(u,v)$ is
invalidating and we need to change the order of vertices in $ORD$. 

Our algorithm to insert an invalidating
edge $(u,v)$ performs various {steps}. Each step
involves visiting an ancestor of $u$ and/or visiting a 
descendant of $v$. 

$\bullet$ Initially $u$ is the only ancestor of $u$ that we know. So
we {\em visit} $u$. We use a Fibonacci heap $F_u$ to store
ancestors of $u$ that we have seen but not yet visited.
For an ancestor $x$ of $u$, {\em visit($x$)} means that for every 
edge $(w,x)$ into $x$, we check if $w$ is already present in 
$F_u$ and if $w$ is not present in $F_u$, we insert $w$ into $F_u$.

$\bullet$  The next ancestor of $u$ that we visit is the vertex with the maximum
$ORD$ label in $F_u$. An {\em extract-max} operation on this F-heap 
(the priority of vertices in $F_u$ is determined by how {\em high} their 
associated label is in $ORD$) determines this vertex $x$.

 $\bullet$ Analogously, we have a Fibonacci heap $F_v$ to store 
descendants of $v$ that we have seen but not yet visited.
For any descendant $y$ of $v$, {\em visit($y$)} 
means that for every edge $(y,z)$ out of $y$, we check if $z$ is
already present in the F-heap $F_v$ and if $z$ is not present in $F_v$,
we insert $z$ into $F_v$. The priority of
vertices in $F_v$ is determined by how {\em low} their associated label 
is in $ORD$. Thus an {\em extract-min} operation on this F-heap determines the
next descendant of $v$ that we visit.

$\bullet$ At the end of each step we check if $Order(x,y)$ is true, where
$x$ is the last extracted vertex from $F_u$ and $y$ is the last
extracted vertex from $F_v$. If $Order(x,y)$ is true (i.e, if $x$ precedes
$y$ in $ORD$), then this is the termination step; 
all the ancestors of $u$ that we visited, call them 
$\{u_0(=u),\ldots,u_k\}$ and 
the descendants of $v$ that we visited, call them $\{v_0(=v),\ldots,v_s\}$,
get reinserted in $ORD$ after $x$ or before $y$,
in the order $u_k, \ldots,u_0,v_0,\ldots,v_s$.
Else, i.e, if $y$ precedes $x$ in $ORD$,
then we 
delete $x$ and $y$ from their current positions in $ORD$ and in the next step
we either visit $x$ or $y$ or both $x$ and $y$.

\medskip

In any step of the algorithm, if $\{u_0,u_1,\ldots,u_r\}$ is the set of
ancestors of $u$ that we have already visited (in this order, so
$\ord(u_r) < \cdots < \ord(u_0)$) in the
previous steps, then the ancestor of $u$
that we plan to visit next is the vertex $x$  with the {\em maximum} $ORD$ label 
that has an edge into a vertex in $\{u_0,u_1,\ldots,u_r\}$. Once we visit $x$,
we would have visited all ancestors of $u$ with $ORD$ labels sandwiched between
$\ord(x)$ and $\ord(u)$.
Similarly, on the side
of $v$, if $v_0, v_1,\ldots,v_{\ell}$ are the descendants of $v$ that we 
have already visited (i.e., $\ord(v_0) < \cdots < \ord(v_{\ell})$), then 
the descendant of $v$ that we plan to visit next is the vertex $y$
with the {\em minimum} $ORD$ label which has an 
edge coming from one of $\{v_0, v_1,\ldots,v_{\ell}\}$. 

When $Order(x,y)$ is true, 
it means that we have discovered { {\em all} descendants of $v$ with $ORD$ label values
between $\ord(v)$ and $i$, and
\em all} ancestors of $u$ with $ORD$ label values
between $i$ and $\ord(u)$ (where $i$ is any value such that $\ord(x) \le i \le \ord(y)$).  
Thus we can relocate
vertices $u_k, \ldots,u_0,v_0,\ldots,v_s$ (in this order) between
$x$ and $y$. It is easy to see that now for every
$(a,b) \in E$, we have that $a$ precedes $b$ in $ORD$.

What remains to be explained is how to make the
choice between the following 3 options in each step:
(i) visit($x$) {\em and} visit($y$),
(ii) only visit($x$), or 
(iii) only visit($y$).

\paragraph{Visit($x$) and/or Visit($y$).}
In order to make the choice between visit($x$) and/or visit($y$), 
let us make the following definitions: 
Let $\anc$ denote the set of ancestors of $u$ that we have already visited
plus the ancestor $x$ that we plan to visit next.
Let $\des$ denote the set of descendants of $v$ that we have already visited
plus the descendant $y$ that we plan to visit next.
Let $m_D$ be the sum of out-degrees of vertices in 
$\des$ and let $m_A$ be the sum of in-degrees of vertices in $\anc$.  

If we were to visit
$x$ in the current step, then the {\em total} work done by us on the side of $u$ so far
would be  $m_A + |\anc|\log n$
(to have examined $m_A$ edges incoming into $\anc$ and for at most $m_A$ insertions
in $F_u$, and to have performed $|\anc|$ many {\em extract-max} operations on $F_u$).
Similarly, if we were to visit
$y$ in the current step, then the total work done by us on the side of $v$ so far 
would be $m_D + |\des|\log n$.
%(to have examined $m_D$ edges going out of $\des$ and for at most $m_D$ insertions
%in $F_v$, and to have performed $|\des|$ many {\em extract-min} operations on $F_v$).

\begin{definition}
\label{balance}
If $m_A \leq m_D \leq m_A + |\anc|\log{n}$ \ or \  
$m_D \leq m_A \leq m_D + |\des|\log{n}$ then we say that
 $m_A$ and $m_D$ are ``balanced'' with respect
to each other. Else we say that they are not balanced 
with respect to each other.
\end{definition}

If $m_A$ and $m_D$ are balanced with respect to each other, then we visit both $x$ and $y$.
Else if $m_A < m_D$, then we visit only $x$, otherwise we visit only $y$.
This is the difference between our algorithm and the algorithm in \cite{KB05} - in the
latter algorithm, either only $x$ is visited or only $y$ is visited unless $m_D = m_A$,
in which case both $x$ and $y$ are visited. In our algorithm we are ready to visit
both $x$ and $y$ more often, that is, whenever $m_A$ and $m_D$ are ``more or less'' equal
to each other.
If we visit both $x$ and $y$,
then the total work done is $m_A + |\anc|\log{n} +  
m_D + |\des|\log{n}$. We can give a good upper bound for this quantity 
using the fact that $m_A$ and $m_D$ are balanced with respect to each other. 
On the other hand, if 
$m_A$ and $m_D$ are not balanced w.r.t. each other, we are not able to give a good upper bound
for the {\em sum} of $(m_A + |\anc|\log{n})$ and $(m_D + |\des|\log{n})$, hence we visit either
$x$ or $y$, depending upon the smaller value in $\{m_A,m_D\}$.

\subsection{The algorithm}
\label{label2}
Our entire algorithm to reorder vertices in $ORD$ upon the insertion of an
invalidating edge $(u,v)$ 
is described as Algorithm~\ref{fig-algo4}. 
This algorithm is basically an implementation  
of what was described in the previous section with a check at the
beginning of every step to see if $m_A$ and $m_D$ are balanced with respect to
each other or not. If they are, then we visit both $x$ and $y$. Else, we visit only
one of them ($x$ if $m_A < m_D$, else $y$). The algorithm maintains
the invariant that the $ORD$ labels of all elements in $\anc$ are higher
than the $ORD$ labels of all elements in $\des$. The termination condition is
determined by $Order(x,y)$ being true, where $x$ is the last extracted vertex from
$F_u$ and $y$ is the last extracted vertex from $F_v$.

For simplicity, in the description of the algorithm 
we assumed that the heaps $F_u$ and $F_v$ remain non-empty 
(otherwise extract-max/extract-min operations would return null values) 
- handling these cases is easy. 
We also assumed that the edges inserted are the edges of a DAG.
Hence we did not perform any cycle detection here. 
(Cycle detection can be easily incorporated, by using 2 flags for each vertex 
that indicate its membership in $F_u$ and in $F_v$.)
When our algorithm terminates, it is easy to see the order of vertices in 
$ORD$ is a valid topological ordering.
We discuss the running time of Algorithm~\ref{fig-algo4} in the next section.

\begin{algorithm}[h]
Initially, $\anc = \{u\}$ and $\des = \{v\}$.

Set $x = u$ and $y = v$. Delete $x$ and $y$ from their current locations in $ORD$.

Set $m_A = u$'s in-degree and $m_D = v$'s out-degree.
%A succinct description of our approach is described below.

\begin{algorithmic}
\WHILE{$\mathsf{TRUE}$}
\IF{$m_A$ and $m_D$ are balanced w.r.t. each other (see Defn.~\ref{balance})}
\STATE{Visit($x$) and Visit($y$).}

\COMMENT{So new vertices get inserted into $F_u$ and into $F_v$.}
\STATE{$x = ${\em extract-max}($F_u$)}
\STATE{$y = ${\em extract-min}($F_v$)}
\IF{$\ord(x) < \ord(y)$}
\STATE{-- insert all elements of $\anc$ (with the same relative order within
themselves) followed by all \hspace*{0.3cm} elements of $\des$ (with the same relative order) 
after $x$ in $ORD$}
%\STATE{-- insert all elements of $\des$ before $y$ in $ORD$.}

%\COMMENT{The order among elements of $\anc$ and among elements of $\des$ is preserved.}
\STATE{{\bf break}}
\COMMENT{This terminates the while loop}
\ELSE
\STATE{Delete $x$ and $y$ from their current positions in $ORD$.}
\STATE{$\anc = \anc \cup \{x\}$ and $\des = \des \cup \{y\}$}
\STATE{$m_A = m_A + x$'s in-degree and $m_D = m_D + y$'s out-degree}
\ENDIF
\ELSIF{$m_A < m_D$}
\STATE{Visit($x$)}
\STATE{$x = ${\em extract-max}($F_u$)}
\IF{$\ord(x) < \ord(y)$}
\STATE{-- insert all elements of $\anc$ followed by all elements of $\des$ after $x$ in $ORD$.}

%\COMMENT{The order among elements of $\anc$ and among elements of $\des$ is preserved.}
\STATE{{\bf break}}
\ELSE
\STATE{Delete $x$ from its current position in $ORD$.} 
\STATE{Set $\anc = \anc \cup \{x\}$ and $m_A = m_A + x$'s in-degree.}
\ENDIF
\ELSE
\STATE{Visit($y$)}
\STATE{Let $y = ${\em extract-max}($F_v$)}
\IF{$\ord(x) < \ord(y)$}
\STATE{-- insert all elements of $\anc$ followed
by all elements of $\des$ before $y$ in $ORD$.}

%\COMMENT{The order among elements of $\anc$ and among elements of $\des$ is preserved.}
\STATE{{\bf break}}
\ELSE
\STATE{Delete $y$ from its current position in $ORD$.} 
\STATE{Set $\des = \des \cup \{y\}$ and $m_D = m_D + y$'s out-degree.}
\ENDIF
\ENDIF
\ENDWHILE
\end{algorithmic}
\caption{Our algorithm to reorder vertices in $ORD$ upon insertion of an invalidating edge $(u,v)$.}
\label{fig-algo4}
\end{algorithm}

\subsubsection{The running time}
\label{new-analysis}
 Let $T(e)$ denote the time taken by Algorithm~\ref{fig-algo4}
to insert an edge $e$. We need to show an upper bound for $\sum_e T(e)$,
where the sum is over all invalidating edges $e$.
For simplicity of exposition, let us define the following modes.
While inserting an edge $(u,v)$,
if a step of our algorithm involved visiting an ancestor of
$u$ {\em and} a descendant of $v$, we say that step was performed in {\em mode~(i)}.
That is, at the beginning of that step, we had $m_A$ and $m_D$ balanced with
respect to each other.
If a step involved visiting only an ancestor of $u$,
then we say that the step was performed in {\em mode~(ii)}, else we say that
the step was performed in {\em mode~(iii)}.

We partition the sum $\sum_e T(e)$
into 2 parts depending upon the {\em mode} of the termination step 
of our algorithm.
Let $S_1 = \sum_e T(e)$ be
the time taken by our algorithm over all those edges $e$
such that the termination step was performed in mode~(i).
Let $S_2 = \sum_e T(e)$ where the sum is over all those
edges $e$ such that the termination step was performed in mode~(ii)
or mode~(iii).
We will show that both  $S_1$ and $S_2$ are $O((m + n\log n)\sqrt{m})$.
These bounds on $S_1$ and $S_2$ will prove Theorem~\ref{thm2} stated
in Section~\ref{intro}.

The following lemma shows the bound on $S_1$. 
We then show an analogous bound on $S_2$.
\begin{lemma}
\label{last-lemma}
$S_1$ is $O((m+n\log n)\sqrt{m})$.
\end{lemma}
\begin{proof}
Let us consider any particular edge $e_i = (u,v)$
such that the last step of Algorithm~\ref{fig-algo4} while inserting $e_i$ 
was performed in mode~(i).
So the termination step involved visiting an ancestor $u_k$ of $u$, extracting the
next ancestor $x$ of $u$, visiting a descendant $v_s$ of $v$, extracting the
next descendant $y$ of $v$ and then checking that $x$ precedes $y$ in $ORD$.

Let the set $\anc = \{u,u_1,\ldots,u_k\}$ and
the set $\des = \{v,v_1,\ldots,v_s\}$. 
Let $m_A$ be the sum of in-degrees of vertices in $\anc$
and let $m_D$ be the sum of out-degrees of vertices in $\des$.
During all the steps of the algorithm, 
we extracted $|\anc|$ many vertices (the vertices $u_1,\ldots,u_k$ and $x$) 
from $F_u$ and $|\des|$ many vertices from $F_v$.
%The total work done on the side of $u$ is essentially examining the endpoints of 
%$m_A$ edges and extracting $|\anc|$ many vertices (the vertices $u_1,\ldots,u_k$ and $x$) 
%from $F_u$. Similarly, the total work
%done on the side of $v$ is examining the endpoints of $m_D$ many
%edges and extracting $|\des|$ many vertices from $F_v$.
So we have
%the total cost of Algorithm~\ref{fig-algo4} to insert $(u,v)$ can be upper bounded by 
%a constant times
$T(e_i)$ is $O(m_A + |\anc|\log n + m_D + |\des|\log n)$.

Since the termination step
was performed in mode~(i), we have that
%\[m_A \le m_D \le m_A + |\anc|\log n \ \ \ \text{or} \ \ \ m_D \le m_A \le m_D + |\des|\log n. \]
$m_A \le m_D \le m_A + |\anc|\log n$ or $m_D \le m_A \le m_D + |\des|\log n$.
Without loss of generality let us assume that $m_A \le m_D \le m_A + |\anc|\log n$. 
Hence $T(e_i)$ can be upper bounded by some constant times
\begin{equation}
\label{EQN0}
m_A + |\anc|\log n + |\des|\log n.
\end{equation}

Let us assume that $|\anc| > |\des|$. (Note that the case $|\anc| < |\des|$ is symmetric
to this and the case $|\anc| = |\des|$ is the easiest.)
Since the termination step was performed in mode~(i), for $|\anc|$ to be larger than $|\des|$,
it must be the case that
at some point in the past, our algorithm to insert $e_i$ was operating in mode~(ii)
and that contributed to accumulating quite a few ancestors of $u$. 
Let step $t$ be the last step that was operated in mode~(ii). 
So at the beginning of step $t$ we had $m'_A + |\anc'|\log n \le m'_D$,
where $\anc'$ was the set of ancestors of $u$
extracted from the F-heap $F_u$ till the beginning of step $t$ and 
$m'_A$ is the sum of in-degrees of vertices in $\anc'$,
and $m'_D$ is the sum of out-degrees of vertices in $\des'$ where
$\des'$ was the set of descendants of $v$
extracted from the F-heap $F_v$ till the beginning of step $t$.
After step~$t$, we never operated our algorithm in mode~(ii). Thus
subsequent to step~$t$ whenever we extracted a vertex from $F_u$,
we also extracted a corresponding vertex from $F_v$. So we have
$|\anc| \le |\anc'| + |\des|$. 
%This yields:
%\begin{eqnarray*}
%|\anc|\log n  & \leq & (|\anc'| + |\des|)\log n \\
%& \leq & |\des|\log n + (m'_{u} + |\anc'|\log n).
%\end{eqnarray*}
Using this inequality in (\ref{EQN0}),
we get that 
\[T(e_i) \le c(m_A + |\des|\log n + |\anc'|\log n), \ \ \ \text{for some constant $c$}.\]

\begin{claim}
We have the following relations: 
\begin{itemize}
\item $m_A^2 \le m_A \cdot m_D \le \Phi(e_i)$

where $\Phi(e_i)$ is the number of pairs of edges $(e,e')$ for which
the relationship $e \leadsto e'$ has started now for the first time due to
the insertion of $e_i$. 
[We say $(a,b) \leadsto (c,d)$ if $b$ is an ancestor of $c$.]

\item $|\des|^2 \le |\anc|\cdot|\des| \le N(e_i)$

where $N(e_i)$ is the number of pairs of vertices $(w,w')$ such that 
$w \leadsto w'$ has started now for the first time due to
the insertion of $e_i$. [We say $w \leadsto w'$ if $w$ is an ancestor of $w'$.]

\item $(|\anc'|\log n)^2 \le (|\anc'|\log n) \cdot m'_D$ $\le \Psi(e_i)\log n$

where $\Psi(e_i)$ is the number of pairs $(w,e) \in V \times E$
for which the relationship $w \leadsto e$ has started now for the first time.
[We say $w \leadsto (a,b)$ if $w$ is an ancestor of $a$.]
\end{itemize}
\end{claim}
\noindent{\em Proof of Claim 1.}
After the insertion of edge $(u,v)$ we have
$e \leadsto e'$ for every edge $e$ incoming into $\anc = \{u,u_1,\ldots,u_k\}$ 
and every edge $e'$ outgoing from $\des =  \{v,v_1,\ldots,v_s\}$.
Prior to inserting $e_i$,
the {\em sink} of each of the $m_A$
edges incoming into $\anc$ has a higher $ORD$ label compared to
the {\em source} of each of the $m_D$ edges outgoing from $\des$ - thus
%Note that prior to inserting $e_i$, each vertex
%in $\anc$ had a higher label in $ORD$ compared to
%each vertex in $\des$. Thus 
we could have had no relation 
of the form $e \leadsto e'$ between the $m_A$ edges
incoming into $\anc$ and the $m_D$ edges outgoing from $\des$.
So $\Phi(e_i) \ge m_Am_D$. 

The above argument also shows that $N(e_i) \ge |\anc|\cdot|\des|$.
We have
$|\anc'|m'_D \le \Psi(e_i)$ because the source of each of these $m'_D$
edges had a lower $ORD$ label than the vertices in $\anc'$ prior to
inserting $e_i$; thus the relation $w \leadsto e'$ for each $w \in \anc'$
and the edges $e'$ ($m'_D$ many of them) outgoing from $\des'$
is being formed for the first time now.  \hspace{\fill} \qed

\medskip 

Now we are ready to finish the proof of Lemma~\ref{last-lemma}.
Corresponding to the insertion of each edge $e_j$ whose termination
step was in mode~(i), the work done by our algorithm is at most
$c(f_j + g_j + h_j)$ where $f_j^2 \le \Phi(e_j)$ and $g_j^2 \le N(e_j)\log^2n$
and $h_j^2 \le \Psi(e_j)\log n$.  In order to bound 
$\sum_j  (f_j + g_j + h_j)$, we use Cauchy's inequality
which states that $\sum_{i=1}^m x_i \le \sqrt{\sum_i x^2_i}\sqrt{m}$, 
for $x_1,\ldots,x_m \in \mathbb{R}$. 
This yields
\begin{eqnarray}
\sum_j  f_j + g_j + h_j & \le & (\sqrt{\sum f_j^2} + \sqrt{\sum g_j^2} +  \sqrt{\sum h_j^2})\sqrt{m}\\
& \le & \left(\sqrt{\sum \Phi(e_j)} + \sqrt{\sum N(e_j)}\log n +  \sqrt{\sum \Psi(e_j)\log n}\right)\sqrt{m}\\
& \le & (m + n\log n + \sqrt{mn\log n})\sqrt{m}.
%& \le & 2(m + n\log n)\sqrt{m}.
\end{eqnarray}
We have $\sum_j \Phi(e_j)$ is at most $m \choose 2$
since each pair of edges $e$ and $e'$ 
can contribute at most 1 to $\sum_j \Phi(e_j)$;
%the first time we have $e \leadsto e'$ or $e' \leadsto e$;
similarly $\sum_j \Psi(e_j)$ is at most $mn$, and $\sum_j N(e_j)$ is at most $n\choose 2$.
This yields Inequality~(9)
from (8).
Since $\sqrt{mn\log n} \le (m + n\log n)/2$,
%(geometric mean $\le$ arithmetic mean), 
%we have thus shown that
%Note that $\sqrt{mn\log n} \le (m + n\log n)/2$ (
this completes the proof that the sum
%it follows that 
$S_1$ is $O((m+n\log n)\sqrt{m})$. 
\end{proof}

Analogous to Lemma~\ref{last-lemma}, we need to show the following lemma in order
to bound the running time of Algorithm~\ref{fig-algo4} by $O((m + n\log n)\sqrt{m})$.
\begin{lemma}
\label{lem9}
$S_2$ is $O((m + n\log n)\sqrt{m})$.
\end{lemma}
\begin{proof}
Recall that $S_2 = \sum T(e)$ where the sum is over all those $e$
such that the termination step of Algorithm~\ref{fig-algo4} was
performed in mode~(ii) or mode~(iii).  
Let us further partition this sum into $\sum_e T(e)$ 
over all those $e$ for which the last step
was performed in mode~(ii) and $\sum_{e'} T(e')$ 
over all those $e'$ for which the last step
was performed in mode~(iii).  The analysis for the second
sum will be entirely symmetric to the first.
We will now bound the first sum.

Let $e_i = (u,v)$ be an edge such that the termination
step of our algorithm was performed in mode~(ii).
Let the set $\anc = \{u,u_1,\ldots,u_k\}$ 
and let the set $\des = \{v,v_1,\ldots,v_s\}$.
Let $m_A$ be the sum of in-degrees of vertices in $\anc$
and let $m_D$ be the sum of out-degrees of vertices in $\des$.
Since the termination step was performed in
mode~(ii), we have $m_A + |\anc|\log n < m_D$.

The work that we did in all the steps while inserting $(u,v)$ 
from the side of the vertex $u$ is $m_A + |\anc|\log n$.
Let step~$t$ be the last step of our algorithm which was operated in
mode~(i) or in mode~(iii). If there was no such step, then the total
work done is at most $m_A + |\anc|\log n$ and it is easy to bound this
using the inequality $m_A + |\anc|\log n < m_D$.
Hence, let us assume that such a step~$t$ did exist and let
$\des'$ be the set of descendants of $v$ at the beginning of step~$t$
and let $m'_D$ be the sum of out-degrees of vertices in $\des'$.
Note that we have $m'_D \le m'_A + |\anc'|\log n$ since this step was
operated in mode~(i) or in mode~(iii).

The total work done from the side of $v$ is $m'_D + |\des'|\log n$.
Thus the total work $T(e_i)$ is 
$O(m_A + |\anc|\log n + m'_{v}  + |\des'|\log n)$.
Using the inequality $m'_D \le m'_A + |\anc'|\log n \le
m_A + |\anc|\log n$, we have $T(e_i)$
upper bounded by a constant times
\begin{equation}
\label{EQN1}
m_A + |\anc|\log n + |\des'|\log n.
\end{equation}

Let us concentrate on the last term
$|\des'|\log n$ in the above sum. Let $\des''$ be the
set of descendants of $v$ at the beginning of the last step when
we ran in mode~(iii). So $m''_D + |\des''|\log n \le m''_A$,
where $m''_D$ is the sum of the out-degrees of vertices in $\des''$,
$\anc''$ is the set of ancestors of $u$ at the beginning of this
step and $m''_A$ is the sum of the in-degrees of vertices in $\anc''$.
After this step, whenever we explored edges on the side of $v$, it was
in mode~(i), thus visiting a descendant of $v$ was always
accompanied by visiting an ancestor of $u$. 
So $|\des'| \le |\des''| + |\anc|$. Substituting this bound in (\ref{EQN1})
we get that
\[ T(e_i) \le c(m_A + |\anc|\log n + |\des''|\log n) \ \ \ \text{for some constant $c$}.\]
We have the following relations 
(see Claim~1 for the definitions of $\Phi(e_i)$ and $\Psi(e_i)$):
\begin{itemize}
\item $(m_A + |\anc|\log n)^2 \le (m_A + |\anc|\log n)m_D$\\ 
\hspace*{3cm} $\le \Phi(e_i) + \Psi(e_i)\log n$.
\item $(|\des''|\log n)^2 \le (|\des''|\log n)m''_A $ \\
\hspace*{2.35cm} $\le \nu(e_i)\log n$

where $\nu(e_i)$ is the number of pairs $(e,z) \in E \times V$ that get ordered with
respect to
each other for the first time now due to the insertion of $e_i$.
\end{itemize}

The proofs of the above relations are analogous to the proofs given in Claim~1 and we
refer the reader to the proof of Claim~1 (in Section~\ref{new-analysis}).

We are now ready to complete the proof of Lemma~\ref{lem9}.
$\sum T(e_i)$ where the sum is over all those $e_i$ whose last step was performed in
mode~(ii) is at most $\sum_i (p_i + q_i)$ where $p_i^2 \le \Phi(e_i) + \Psi(e_i)\log n$
and $q_i^2 \le \nu(e_i)\log n$. Note that $\sum_i \nu(e_i)$ is at most $mn$ since 
each pair $(e,z) \in E\times V$ can contribute at most 1 to $\sum_i \nu(e_i)$.
Using Cauchy's inequality, we have
\begin{eqnarray*}
\sum_i (p_i + q_i) & \le & (\sqrt{\sum p_i^2} + \sqrt{\sum q_i^2})\sqrt{m}\\
& \le & \left(\sqrt{\sum \Phi(e_i) + \sum \Psi(e_i)\log n} + \sqrt{\sum \nu(e_i)\log n}\right)\sqrt{m}\\
& \le & (\sqrt{m^2 + mn\log n} + \sqrt{mn\log n})\sqrt{m}\\
& \le & (m + 2\sqrt{mn\log n})\sqrt{m}
\end{eqnarray*}

Analogously, we can show that $\sum T(e_i)$ where the sum is over all those 
$e_i$ whose last step was performed in mode~(iii) is at most
$O((m + \sqrt{mn\log n})\sqrt{m})$. Thus $S_2$ is $O((m + \sqrt{mn\log n})\sqrt{m})$.
Since $\sqrt{mn\log n} \le (m + n\log n)/2$ (geometric mean is at most the arithmetic mean),
we have $S_2$ is $O((m + n\log n)\sqrt{m})$.
\end{proof}

\paragraph{Conclusions.}We considered the problem of 
maintaining the topological order of a directed acyclic graph on $n$
vertices under an online edge insertion sequence of $m$ edges. This problem has
been well-studied and the previous best upper
bound for this problem was 
$O(\min\{m^{3/2}\log n, \ m^{3/2}+n^2\log n, \ n^{2.75}\})$.
Here we showed an improved upper bound of 
$O(\min(n^{5/2}, (m+n\log n)\sqrt{m}))$ for this problem.

\medskip

\paragraph{\em Acknowledgments.} We are grateful to Deepak Ajwani and Tobias
Friedrich for their helpful feedback.
%We are grateful to Deepak Ajwani for 
%suggesting this problem to us and 
%for his helpful feedback, and to Tobias Friedrich for providing us their code
%and helping us with it.

%\newpage

\bibliographystyle{plain}

\end{document}